\documentclass[10pt,final,technote]{IEEEtran}

\newcommand{\SQS}{\ensuremath{\mathrm{SQS}}}
\newcommand{\STS}{\ensuremath{\mathrm{STS}}}
\newcommand{\Aut}{\ensuremath{\mathrm{Aut}}}

\newcommand{\supp}{\ensuremath{\mathrm{supp}}}
\newcommand{\wt}{\ensuremath{\mathrm{wt}}}
\newcommand{\F}{\ensuremath{\mathbf{F}}}
\newtheorem{Thm}{Theorem}
\newtheorem{Lem}{Lemma}
\newtheorem{Cor}{Corollary}

\begin{document}

\title{Reconstructing Extended Perfect Binary One-Error-Correcting
Codes from Their Minimum Distance Graphs}
\author{Ivan Yu. Mogilnykh,
Patric R. J. {\"O}sterg{\aa}rd,
Olli Pottonen, Faina I. Solov'eva
\thanks{This work was supported in part by the Graduate School in
Electronics, Telecommunication and Automation and by
the Academy of Finland, Grant Numbers 107493 and 110196.}%
\thanks{I. Yu. Mogilnykh and F. I. Solov'eva are with the
Sobolev Institute of Mathematics and Novosibirsk State University,
Novosibirsk, Russia (e-mail: ivmog84@gmail.com, sol@math.nsc.ru).}
\thanks{P. R. J. {\"O}sterg{\aa}rd and O. Pottonen are with the
Department of Communications and Networking, Helsinki University of
Technology TKK, P.O.\ Box 3000, FI-02015 TKK, Finland (e-mail:
patric.ostergard@tkk.fi, olli.pottonen@tkk.fi).}}

\maketitle
\begin{abstract}
  The minimum distance graph of a code has the codewords as vertices
  and edges exactly when the Hamming distance between two codewords
  equals the minimum distance of the code. A constructive proof for
  reconstructibility of an extended perfect binary one-error-correcting
  code from its minimum distance graph is presented. Consequently,
  inequivalent such codes have nonisomorphic minimum distance graphs.
  Moreover, it is shown that the automorphism
  group of a minimum distance graph is isomorphic to that of the
  corresponding code.
\end{abstract}

\section{Introduction}

A binary code of length $n$ is a subset of $\F_2^n$, where $\F_2 =
\{0, 1\}$ is the field of two elements. Throughout this work, we use
``code'' in the meaning of ``binary code''.  The \emph{support}
$\supp(\mathbf{x})$ of a word $\mathbf{x} = (x_1,x_2,\ldots ,x_n)$ is
the set of its nonzero coordinates, the \emph{weight}
$\wt(\mathbf{x})$ of $\mathbf{x}$ is the number of nonzero
coordinates, and the \emph{Hamming distance} $d_H(\mathbf{x},
\mathbf{y})$ is the number of coordinates in which the words
$\mathbf{x}$ and $\mathbf{y}$ differ. Formally, $\supp(\mathbf{x}) :=
\{i : x_i = 1\}$, $\wt(\mathbf{x}) := |\supp(\mathbf{x})|$ and
$d_H(\mathbf{x}, \mathbf{y}) := \wt(\mathbf{x} - \mathbf{y})$.

The \emph{minimum distance} of a code is the minimum Hamming distance
between any pair of distinct codewords.  For a code with minimum
distance $d$, the balls of radius $r = \lfloor (d-1)/2 \rfloor$
centered around the codewords are nonintersecting and such a code is
called an \emph{$r$-error-correcting code.} If the balls cover the
entire ambient space, the code is called \emph{perfect}, or more
specifically, \emph{$r$-perfect}. With one exception (the binary Golay
code), all nontrivial perfect binary codes have $d=3$, $n = 2^m-1$.

A permutation $\pi$ acts on a codeword by permuting the coordinates.
A pair $(\pi, {\bf z})$ acts on a codeword ${\bf x}$ as $(\pi, {\bf
  z})({\bf x}) = {\bf z} + \pi({\bf x})$.  Two codes are
\emph{equivalent} if the action of such a pair on the codewords of one
code produces the codewords of the other.  The set of all such pairs
that map a code onto itself form the \emph{automorphism group} of the
code.

A \emph{Steiner system} $S(t, k, v)$ is a set of $v$ \emph{points}
together with a collection of \emph{blocks}, each consisting of $k$
points, such that any $t$ points occur in a unique block. The Steiner
systems $S(2, 3, v)$ and $S(3, 4, v)$ are called \emph{Steiner triple
  systems} and \emph{Steiner quadruple systems}, respectively, or
$\STS(v)$ and $\SQS(v)$ for short. If $C$ is a $1$-perfect code of
length $n$ and $\mathbf{x} \in C$, then the blocks $\{\supp(\mathbf{x}
- \mathbf{y}) : d_H(\mathbf{x}, \mathbf{y}) = 3,\ \mathbf{y} \in C\}$
form an $\STS(n)$, called the \emph{neighborhood STS} of $\mathbf{x}$.
Similarly, if $C$ is an extended $1$-perfect code, then each
$\mathbf{x} \in C$ has an \emph{neighborhood SQS} with the block set
$\{\supp(\mathbf{x} - \mathbf{y}) : d_H(\mathbf{x},\ \mathbf{y}) = 4,\
\mathbf{y} \in C\}$.  The \emph{block graph} of an $S(t, k, v)$ has
the blocks of the design as vertices, with edges incident to
intersecting blocks.

The \emph{minimum distance graph} of a code with minimum distance $d$
has the codewords as vertices and edges between codewords with Hamming
distance $d$.  In the rest of the paper we consider such minimum
distance graphs.  Note that the distance between codewords is then the
distance between the corresponding vertices in the graph; this is not
to be confused with the Hamming distance.

Phelps and LeVan~\cite{PL99} asked whether 1-perfect codes with
isomorphic minimum distance graphs are always equivalent, and this
question was answered in the affirmative by Avgustinovich \cite{A01},
building on earlier work by Avgustinovich and others
\cite{A95,SAHH98}; in fact, the result was announced already in
\cite{A95} for lengths $n \geq 31$, but without details.

We start off in Section~\ref{sec:iso} by finalizing a proof that
extended 1-perfect codes with isomorphic minimum distance graphs are
equivalent for $n \geq 256$. The detailed treatment in the rest of the
paper makes it possible to handle codes of shorter lengths. We prove
in Section~\ref{sec:extended} the stronger result that any extended
1-perfect code can be reconstructed from its minimum distance graph,
and, in Section~\ref{sec:perfect}, show how this implies an analogous
result for 1-perfect codes. In Section~\ref{sec:autom} we prove that
the automorphism groups of these codes are isomorphic to the
automorphism groups of their minimum distance graphs for lengths $n
\geq 15$.  Section~\ref{sec:concl} concludes the paper.

\section{Code Isometry and Equivalence}

\label{sec:iso}

A bijection $I: C_1 \rightarrow C_2$ is called an \emph{isometry} if
$d_H(\mathbf{x},\mathbf{y}) = d_H(I(\mathbf{x}),I(\mathbf{y}))$ for
all $\mathbf{x},\mathbf{y} \in C_1$. Moreover, such a mapping is a
\emph{weak isometry} if $d_H(\mathbf{x},\mathbf{y}) = d$ iff
$d_H(I(\mathbf{x}),I(\mathbf{y})) = d$, where $d$ is the minimum
distance of the codes $C_1$ and $C_2$.

We may now rephrase the question by Phelps and LeVan \cite{PL99} in
the defined terms: Are weakly isometric 1-perfect codes always
equivalent?  The idea of the proof completed in \cite{A01} is to
combine a proof that weakly isometric such codes are isometric with a
proof (from \cite{A95,SAHH98}) that isometric such codes are
equivalent.  We may act analogously for extended 1-perfect codes, and
use a result from \cite{SA03} that isometric such codes are equivalent
for lengths $n \geq 256$.  Then it only remains to prove that weakly
isometric codes are isometric, which can be done for arbitrary
lengths.

\begin{Thm}
\label{thm:wi}
Weakly isometric extended 1-perfect codes are isometric.
\end{Thm}

\begin{proof}
  We show that one is able to deduce the Hamming distance between any
  two codewords, given the minimum distance graph. Consider an
  arbitrary codeword $\mathbf{x}$. The codewords $\mathbf{y}$ with
  $d_H(\mathbf{x},\mathbf{y})=4$ are given by the minimum distance
  graph. Having identified all codewords $\mathbf{y}$ with
  $d_H(\mathbf{x},\mathbf{y})\leq i$, we need to distinguish between
  the cases $d_H(\mathbf{x}, \mathbf{z}) = i+2$ and $d_H(\mathbf{x},
  \mathbf{z}) = i+4$ for a codeword $\mathbf{z}$ in order to proceed
  with induction.  If $\mathbf{z}$ has a neighbour $\mathbf{v}$ with
  $d_H(\mathbf{x}, \mathbf{v}) = i-2$, then $d_H(\mathbf{x},
  \mathbf{z}) = i+2$. All remaining codewords $\mathbf{z}$ with
  $d_H(\mathbf{x}, \mathbf{z}) = i+2$ have ${i+2 \choose 3}$
  neighbours that are at Hamming distance $i$ from $\mathbf{x}$,
  whereas those codewords $\mathbf{z}$ with $d_H(\mathbf{x},
  \mathbf{z}) = i+4$ have at most ${i+4 \choose 3}/4$ such neighbors
  (consider respectively the triples and quadruples of
  $\supp(\mathbf{x}-\mathbf{z})$ in the neighborhood SQS of
  $\mathbf{z}$). For $i \ge 4$ we have ${i+2 \choose 3} > {i+4 \choose
    3}/4$.
\end{proof}

\begin{Thm}
  Weakly isometric extended 1-perfect codes are equivalent for lengths
  $n \geq 256$.
\end{Thm}

\begin{proof}
  Follows from Theorem~\ref{thm:wi} and \cite{SA03}.
\end{proof}

\section{Reconstructing Extended 1-Perfect Codes}

\label{sec:extended}

A \emph{clique} in a graph is a set of mutually adjacent vertices.
The idea of utilizing maximum cliques in reconstruction has earlier
been used by Spielman~\cite{S96}; see also~\cite{KO04}. It follows
from a result by Rands \cite{R82} that the maximum cliques in the
block graph of a Steiner system can be used to identify the points of
the design whenever the number of points ($v$) exceeds a certain value
that depends only on the parameters $k$ and $t$. Unfortunately, the
bound derived in \cite{R82} for the threshold value is too large for
the smallest cases that we want to handle, so we need to carry out a
more detailed treatment.

In the preparation for a reconstructibility proof for extended
1-perfect codes, Theorem~\ref{Thm:extended}, we prove three lemmata.

\begin{Lem}\label{lem:dist6}
  The codewords with Hamming distance $6$ can be recognized from the
  minimum distance graph of an extended $1$-perfect code.
\end{Lem}
\begin{proof}
  Follows from the proof of Theorem~\ref{thm:wi}.
\end{proof}

\begin{Lem}\label{lem:cliquesize}
  If $Q$ is a clique in the block graph of an $\SQS(v)$, $v \ge 16$,
  such that there is no point that occurs in every block of $Q$, then
  $|Q| < (v-1)(v-2)/6$.
\end{Lem}

\begin{proof}
  Consider a clique $Q$ such that no point occurs in every block of
  $Q$.  First note that any pair of points is contained in (v-2)/2
  blocks of an SQS(v) and therefore in at most (v-2)/2 blocks of Q.

  We consider the size of a nonempty $Q$ in three separate cases.
\begin{enumerate}
\item\label{it:sp} There is a point $x$ that occurs in every block of
  $Q$ except one:

  Assume that $x \not\in \{a, b, c, d\} \in Q$.  Since $Q$ is a clique
  in the block graph, every block of $Q$ containing $x$ contains at
  least one of the pairs $\{x, a\}, \{x, b\}, \{x, c\}, \{x,
  d\}$. From the fact that each pair occurs in at most $(v-2)/2$
  blocks, it follows that $|Q| \le 4(v-2)/2 + 1 = 2v-3$.

\item There is a pair of points $\{x, y\}$ that intersects every block
  of $Q$, but no point occurs in $|Q|-1$ blocks:

  There are at least two blocks that do not contain $x$; let $B_1$ and
  $B_2$ be two such blocks. Since $x \not\in B_1$ and $x \not\in B_2$,
  by the assumption $y \in B_1\cap B_2$.  If $|B_1\cap B_2| = 2$, $B_1
  = \{y,a,b,c\}$ and $B_2 = \{y,a,d,e\}$ with distinct elements
  $a,b,c,d,e$. Any block that contains $x$ but not $y$ must contain
  either $a$ (there are at most $(v-2)/2$ such blocks), or $b$ and $d$
  (at most 1), or $b$ and $e$ (at most 1), or $c$ and $d$ (at most 1),
  or $c$ and $e$ (at most 1), so there are at most $(v-2)/2 + 4$
  blocks that contain $x$ but not $y$. On the other hand, if $|B_1\cap
  B_2| = 1$, then $B_1 = \{y,a,b,c\}$ and $B_2 = \{y,d,e,f\}$, and we
  get at most 9 blocks containing $x$ and intersecting $B_1$ and
  $B_2$, one for each pair with one element taken from $\{a,b,c\}$ and
  the other from $\{d,e,f\}$.  An upper bound for the number of blocks
  containing $x$ but not $y$ is then $\max\{9,v/2+3\} = v/2+3$ as
  $v\ge 16$.

  By the same argument there are at most $v/2+3$ blocks that contain
  $y$ but not $x$. Finally, at most $(v-2)/2$ blocks contain both $x$
  and $y$, so $|Q| \le (v-2)/2 + 2(v/2+3) = 3v/2 + 5$.

\item For every pair of points there is a block of $Q$ that does not
  intersect the pair:

  (Note that in this case no point occurs in $|Q|-1$ blocks).  Any
  pair of points may occur in at most $4$ blocks of $Q$, since $Q$
  contains a block $B$ that does not intersect the pair, and each
  block that contains the pair also contains a point of $B$.

  Take any point $x$. There are at least two blocks that do not
  contain $x$. If these blocks intersect in two points, say $B_1 =
  \{a,b,c,d\}$ and $B_2 = \{a,b,e,f\}$, we get that each block
  containing $x$ must contain $a$ (at most 4 blocks), $b$ (at most 4),
  $c$ and $e$ (at most 1), $c$ and $f$ (at most 1), $d$ and $e$ (at
  most 1), or $d$ and $f$ (at most 1), giving a total of at most 12
  blocks. Similarly, for the situation with one point in the
  intersection, $B_1 = \{a,b,c,d\}$, $B_2 = \{a,e,f,g\}$, we get an
  upper bound of $4+3^2=13$ blocks.  Thus any point occurs in at most
  $13$ blocks.

  If each point occurs in at most $8$ blocks, we have $|Q| \le
  1+4(8-1) = 29$ as any block must intersect a given block. Assuming
  that there is a point $x$ occurring in at least $9$ blocks, and
  considering blocks containing $x$ and intersecting a block $B$ that
  does not contain $x$, we get by the pigeonhole principle that some
  pair $\{x,y\}$ with $y \in B$ must occur it at least $3$ blocks. Now
  consider a block $\{x,y,a,b\} \in Q$.  There are at most $2\cdot 13
  - 3 = 23$ blocks that intersect $\{x,y\}$. By considering blocks
  intersecting three blocks $\{x,y,a,b\}$, $\{x,y,c,d\}$, and
  $\{x,y,e,f\}$, one obtains that a block that does not intersect
  $\{x,y\}$ must contain one of $2^3 = 8$ sets, $\{a,c,e\}$,
  etc. Moreover, since no two blocks may intersect in three points,
  their total number is at most 8. Summing up the number of blocks
  that intersect $\{x,y\}$ and those that do not, we get that $|Q| \le
  23+8=31$.
\end{enumerate}

Combining the results above, we conclude that $|Q| \le \max(2v-3,
3v/2+5, 31) < (v-1)(v-2)/6$ when $v \ge 16$, and the result follows.
\end{proof}

\begin{Lem}\label{lem:SQS}
  For $v \ge 16$, an $\SQS(v)$ can be reconstructed (up to
  isomorphism) from its block graph.
\end{Lem}
\begin{proof}
  The blocks that contain a specified point form a clique of size
  $(v-1)(v-2)/6$, and the clique corresponds to the blocks of a
  derived $\STS(v-1)$. By Lemma \ref{lem:cliquesize}, other types of
  cliques cannot be this large, so an $\SQS(v)$ can be reconstructed
  from its block graph by finding maximum cliques and identifying them
  with points.
\end{proof}

We have now made all preparations for the main result.

\begin{Thm}\label{Thm:extended}
  An extended $1$-perfect code can be reconstructed (up to
  equivalence) from its minimum distance graph.
\end{Thm}
\begin{proof}
  For lengths $n \le 8$ the claim is trivial as these codes are
  unique, so we assume that $n \ge 16$.

  Identify an arbitrary vertex with the all-zero codeword ${\bf 0}$.
  By Lemma~\ref{lem:dist6} we can construct the block graph of the
  neighborhood SQS of $\mathbf{0}$, and by Lemma~\ref{lem:SQS} the
  neighborhood SQS itself. Now we have reconstructed all codewords
  with weight at most $4$.

  The codewords with weight $6$ can be recognized by
  Lemma~\ref{lem:dist6} and reconstructed as follows.  Assume that
  $\mathbf{x}$ is such a codeword.  If $x_i = 1$, then $\mathbf{x}$
  has ${5 \choose 2} = 10$ neighbors $\mathbf{y}$ with $y_i = 1$,
  $\wt(\mathbf{y}) = 4$; if $x_i = 0$, then an upper bound for the
  number of such neighbors is given by the maximum size of a code of
  length 6, constant weight 3, and minimum distance 4, which is 4.

  We proceed with induction on the weight of codewords. Assume that we
  have reconstructed all codewords with weight at most $w$, $w \ge 6$,
  and let $\mathbf{x}$ be a codeword with weight $w$.

  For each coordinate $r$ there is a set $\{i, j, k\} \subset
  \supp(\mathbf{x})$ such that $\{i, j, k, r\}$ is not a block of the
  neighborhood SQS of $\mathbf{x}$.  Accordingly, $\mathbf{x}$ has
  three distinct neighbors $\mathbf{v}, \mathbf{y}, \mathbf{z}$ such
  that $\{r, i, j\} \subset \supp(\mathbf{x}-\mathbf{v})$, $\{r, i,
  k\} \subset \supp(\mathbf{x}-\mathbf{y})$, and $\{r, j, k\} \subset
  \supp(\mathbf{x}-\mathbf{z})$.  Each of $\mathbf{v}, \mathbf{y},
  \mathbf{z}$ has weight at most $w$, and hence those codewords are
  known. Furthermore,
  \begin{equation} \label{eq:r}
    \{r\} = \supp(\mathbf{x}-\mathbf{v})\cap \supp(\mathbf{x}-\mathbf{y})
            \cap \supp(\mathbf{x}-\mathbf{z}).
  \end{equation}

  Consider the block graph of the neighborhood SQS of $\mathbf{x}$,
  and the maximum cliques from Lemma~\ref{lem:SQS}. Using~(\ref{eq:r})
  we can recognize the clique corresponding to the coordinate $r$. Now
  we know which neighbors of $\mathbf{x}$ differ from $\mathbf{x}$ in
  that coordinate. By repeating this for every $r$ we can reconstruct
  the codewords corresponding to the neighbors of $\mathbf{x}$, and
  the result follows as each codeword $\mathbf{y}$ that is not the
  all-one word has a neighbor of weight $\wt(\mathbf{y})-2$.
\end{proof}

\begin{Cor}
  Weakly isometric extended 1-perfect codes are equivalent.
\end{Cor}

\section{Reconstructing 1-Perfect Codes}

\label{sec:perfect}

We will handle the problem of reconstructing a $1$-perfect code from
its minimum distance graph by reducing it to the problem of
reconstructing an extended $1$-perfect code from its minimum distance
graph.

\begin{Lem}\label{lem:dist4}
  The codewords with Hamming distance $4$ can be recognized from the
  minimum distance graph of a $1$-perfect code.
\end{Lem}
\begin{proof}
  If codewords $\mathbf{x}, \mathbf{y}$ have Hamming distance $4$,
  then their neighborhoods intersect in ${4 \choose 2} = 6$ vertices,
  since for any two coordinates of $\supp(\mathbf{x} + \mathbf{y})$
  there is one neighbor of $\mathbf{x}$ which differs from
  $\mathbf{x}$ in those coordinates.

  If the codewords are at distance $6$, size of the intersection of
  their neighborhoods is at most $4$ (attained by a Pasch
  configuration), and for other distances the neighborhoods do not
  intersect.
\end{proof}

\begin{Thm}\label{Thm:perfect}
  A $1$-perfect code can be reconstructed (up to equivalence) from its
  minimum distance graph.
\end{Thm}
\begin{proof}
  Add new edges between codewords with Hamming distance $4$ (Lemma
  \ref{lem:dist4}).  This gives the minimum distance graph for the
  extended code (obtained by adding a parity coordinate).  Using
  Theorem~\ref{Thm:extended} we can reconstruct the extended code. All
  codewords connected by new edges in the first step of the proof
  differ in the parity coordinate, which can thereby be detected. By
  puncturing in the parity coordinate we get the 1-perfect code.
\end{proof}

\begin{Cor}
  Weakly isometric 1-perfect codes are equivalent.
\end{Cor}

\section{Automorphism Groups}\label{sec:autom}

The fact that the automorphism group of a 1-perfect code is isomorphic
to the automorphism group of its minimum distance graph (for lengths
$n \geq 15$) follow implicitly from \cite{A01,A95,SAHH98}, and the
analogous result for extended 1-perfect codes (for lengths $n \geq
256$) from \cite{SA03} combined with Theorem \ref{thm:wi}.  The
current study enables direct and concise proofs of these facts
(expanded to lengths $n \geq 16$ for extended codes).

\begin{Thm}\label{thm:extautom}
  The automorphism group of an extended $1$-perfect code of length $n
  \geq 16$ is isomorphic to that of its minimum distance graph.
\end{Thm}
\begin{proof}
  The automorphisms of the code can be mapped to automorphisms of the
  graph in the obvious fashion. Using the construction of
  Theorem~\ref{Thm:extended}, this homomorphism can be inverted; more
  specifically, we get an automorphism of the code by checking how
  $\alpha \in \Aut(G)$ maps the codeword $\mathbf{0}$ and the cliques
  used in the construction.
\end{proof}

The result for 1-perfect codes now follows easily.

\begin{Thm}\label{thm:prfaut}
  The automorphism group of a $1$-perfect code of length $n \geq 15$
  is isomorphic to that of its minimum distance graph.
\end{Thm}
\begin{proof}
  We use the construction from Theorem~\ref{Thm:perfect}. Assume that
  extending a $1$-perfect code $C$ with a parity coordinate yields the
  code $C'$. Now $\Aut(C)$ is the subgroup of $\Aut(C')$ that
  stabilizes the parity coordinate. Similarly, if $G$ is the minimum
  distance graph of $C$ and $G'$ the graph constructed in
  Theorem~\ref{Thm:perfect}, $\Aut(G)$ is the subgroup of $\Aut(G')$
  that stabilizes the new edges setwise. By Theorem~\ref{thm:extautom}
  these subgroups are isomorphic, and hence $\Aut(C) \cong \Aut(G)$ as
  well.
\end{proof}

\section{Conclusions}\label{sec:concl}

The result that $1$-perfect and extended $1$-perfect codes can be
reconstructed from their minimum distance graphs is not only of
theoretical interest but also has practical implications.  Several
methods have been used for deciding equivalence of (extended)
$1$-perfect codes~\cite{PL99,OP09,P00}---the most straightforward
method of representing the codes as graphs and deciding isomorphism of
these graphs is rather inefficient~\cite{P00}.  The results obtained
imply that this problem reduces to determining whether their minimum
distance graphs are isomorphic.

\end{document}